\documentclass[draft]{IEEEtran}
\onecolumn
\usepackage{amsmath,amssymb,amsthm}
\usepackage[final,hyperfootnotes=false]{hyperref}
\usepackage{graphicx}
\usepackage{epstopdf}
\usepackage{dblfloatfix}
\usepackage{tikz}
\usetikzlibrary{arrows,positioning,fit,backgrounds,shapes}
\usepackage{float}

\newtheorem{thm}{Theorem}

\newtheorem{lem}[thm]{Lemma}
\newtheorem{prop}[thm]{Proposition}
\newtheorem{example}[thm]{Example}
\theoremstyle{definition}
\newtheorem{defn}{Definition}

\theoremstyle{remark}
\newtheorem{rem}{Remark}

\title{Improving the Lower Bound for the Union-closed Sets 
Conjecture via Conditionally IID Coupling}
\author{%
  \IEEEauthorblockN{Jingbo Liu}\\
  \IEEEauthorblockA{Department of Statistics,
   University of Illinois, Urbana-Champaign\\
                                    Email: jingbol@illinois.edu}
}

\date{December 2021}

\date{May 2023}

\begin{document}
\maketitle
\begin{abstract}
Recently, Gilmer proved the first constant lower bound for the union-closed sets conjecture via an information-theoretic argument.
The heart of the argument is an entropic inequality involving the OR function of two i.i.d.\ binary vectors, and the best constant obtainable through the i.i.d.\ coupling is $\frac{3-\sqrt{5}}{2}\approx0.38197$.
Sawin demonstrated that the bound can be strictly improved by considering a convex combination of the i.i.d.\ coupling and the max-entropy coupling, and the best constant obtainable through this approach is around 0.38234, as evaluated by Yu and Cambie. 
In this work we show analytically that the bound can be further strictly improved by considering another class of coupling under which the two binary sequences are i.i.d.\ conditioned on an auxiliary random variable.
We also provide a new class of bounds in terms of finite-dimensional optimization.
For a basic instance from this class, analysis assisted with numerically solved 9-dimensional optimization suggests that the optimizer assumes a certain structure.
Under numerically verified hypotheses, the lower bound for the union-closed sets conjecture can be improved to approximately 0.38271, a number that can be defined as the solution to an analytic equation.
\end{abstract}

\section{Introduction}

The union-closed sets conjecture, usually credited to Frankl, is a well-known open problem in combinatorics
which has a simple statement: 
for any nonempty union-closed family of subsets of $[n]:=\{1,2,\dots,n\}$, there exists $i\in[n]$ that belongs to at least half of these subsets \cite{bruhn2015journey}.
Combinatorial arguments of Knill \cite{knill1994graph}
and W\'ojick\cite{wojcik1999union}
proved weaker versions of the conjecture where the proportion one half is replaced by a lower bound depending on the size of the family.
The first constant lower bound was achieved recently by Gilmer \cite{gilmer2022constant} through an information theoretic argument. 
Let $\bar{s}:=1-s$ for any $s\in[0,1]$, and let $h(s):=s\log_2\frac1{s}+\bar{s}\log_2\frac1{\bar{s}}$
denote the binary entropy function.
The crux of argument is the following information-theoretic inequality:
\begin{prop}\label{prop1}
Let $S$ be a random variable in $[0,1]$ satisfying $\mathbb{E}[S]=u$, and let $T$ be an i.i.d.\ copy of $S$. Then
\begin{align}
\mathbb{E}[h(\bar{S}\bar{T})]
\ge \mathbb{E}[h(S)]\cdot\left\{
\begin{array}{cc}
\frac{h(2u-u^2)}{h(u)},
\quad 
u\le \frac{3-\sqrt{5}}{2}
\\
(1-u)\frac{2}{\sqrt{5}-1},
\quad u\ge \frac{3-\sqrt{5}}{2}
\end{array}
\right..
\label{e1}
\end{align}
\end{prop}
Gilmer's original paper \cite{gilmer2022constant} established a similar bound with a suboptimal constant; 
the sharp bound \eqref{e1} was established by \cite{chase2022approximate}\cite{alweiss2022improved}\cite{sawin2022improved}.
Using Proposition~\ref{prop1}, a weak form of the union-closed sets conjecture follows with constant $\frac{3-\sqrt{5}}{2}$:
Indeed, as in \cite{gilmer2022constant}, one may consider $X^n:=\{X_1,\dots,X_n\}$ and $Y^n$ two independent and identically distributed (i.i.d.) binary vectors, such that $\{i\in[n]\colon X_i=1\}$ is equiprobably distributed on the union-closed family of of sets.
Then, denoting by $H(\cdot|\cdot)$ the conditional Shannon entropy in bits, we have
\begin{align}
H(X^n\vee Y^n)&=\sum_{i=1}^n H(X_i\vee Y_i|X^{i-1}\vee Y^{i-1})
\\
&\ge \sum_{i=1}^n H(X_i\vee Y_i|X^{i-1}, Y^{i-1}),
\label{e_3}
\end{align}
where $\vee$ denotes the elementwise max, whereas 
\begin{align}
H(X^n)&=\sum_{i=1}^nH(X_i|X^{n-1})
\\
&=\sum_{i=1}^nH(X_i|X^{i-1},Y^{i-1}).
\end{align}
Thus, if $\mathbb{E}[X_i]<\frac{3-\sqrt{5}}{2}$ for all $i$, then setting $S=\mathbb{E}[X_i|X^{i-1}]$ and $T=\mathbb{E}[Y_i|Y^{i-1}]$, 
we obtain from Proposition~\ref{prop1} that $H(X^n\vee Y^n)>H(X^n)$,
so that the family cannot be union-closed as the equiprobable distribution maximizes the entropy for a given support.

Clearly, for this argument to work we only used the fact that $\max_{P_{X^nY^n}}H(X^n\vee Y^n)-H(X^n)>0$, where the max is over $P_{X^nY^n}$ under which $P_{X^n}=P_{Y^n}$.
This inequality appears very similar to the entropic formulation of a general version of the  reverse Brascamp-Lieb inequality, for which a tensorization property is known (see \cite{liu2017information}\cite{liu2018forward} for the max entropy version and 
\cite{anantharam2022unifying} for the independent coupling case);
the reason why we cannot apply tensorization of Brascamp-Lieb here (and hence proving the union-closed sets conjecture by solving a simple $n=1$ case) is that $P_{X^n}=P_{Y^n}$ plays an essential role.

It is tempting to strengthen Gilmer's lower bound by considering other coupling of $X_i$ and $Y_i$ so that $P_{X_i|X^{i-1}Y^{i-1}}=P_{X_i|X^{i-1}}$, 
$P_{Y_i|X^{i-1}Y^{i-1}}=P_{Y_i|Y^{i-1}}$ is still true, yet $H(X_i\vee Y_i|X^{i-1},Y^{i-1})$ becomes larger.
The main challenge, however, is that $S:=\mathbb{E}[X_i|X^{i-1}]$ and $T:=\mathbb{E}[Y_i|Y^{i-1}]$ will then have more complicated dependence structure, whereas the validity of \eqref{e1} relies strongly on the independence of $S$ and $T$.
Sawin \cite{sawin2022improved} proved that by taking a convex combination of the i.i.d.\ coupling and a coupling that maximizes $H(X_i\vee Y_i|X^{i-1},Y^{i-1})$, one can strictly improve the best lower bound $\frac{3-\sqrt{5}}{2}$ obtained by the i.i.d.\ coupling.
The best lower bound obtained by this approach was evaluated by Yu \cite{yu2023dimension} and Cambie \cite{cambie2022better}, the heart of which is the following:
\begin{prop}\label{prop2}
For $c^*$ and $\alpha^*$ that can be analytically defined (see below this proposition), the following is true:
For any $c<c^*$,
there exists $C>1$ such that 
\begin{align}
\bar{\alpha}^*\mathbb{E}[h(\bar{S}\bar{T})]+\alpha^*\mathbb{E}[h(S\vee R\vee\min(S+R,1/2))]
\ge C\mathbb{E}[h(S)]
\label{e_6}
\end{align}
whenever $P_{SR}$ is a symmetric distribution, $P_{ST}=P_SP_T$, $P_S=P_T$, and $\mathbb{E}[S]\le c$. 
\end{prop}
Note that $h(S\vee R\vee\min(S+R,1/2))$ is the maximum $H(X_i\vee Y_i|X^{i-1},Y^{i-1})$ for given $S$ and $T$.
Intuitively, \eqref{e_6} can hold for some $c^*>\frac{3-\sqrt{5}}{2}$ because the optimizer of the ratio of expected entropies in \eqref{e1} is close to a point mass at $\frac{3-\sqrt{5}}{2}$ when $u$ is close to $\frac{3-\sqrt{5}}{2}$, so when $\alpha^*$ is small the optimizer for \eqref{e_6} must be close to that point mass (in the sense of weak convergence), but the max entropy coupling produces strictly larger $\mathbb{E}[h(S\vee R\vee\min(S+R,1/2))]$ in that case.

To find the best values of $c^*$ and $\alpha^*$ in Proposition~\ref{prop2}, one can first use Krein-Milman to reduce the optimization to a finite (five) dimensional one \cite{yu2023dimension}.
Further cardinality reduction with mass moving arguments were given in  
\cite{cambie2022better}. 
Then the independent numerical optimization by Yu and Cambie confirmed that the optimizer is 
\begin{align}
P_{SR}(1,b)&=P_{SR}(b,1)=a;
\label{e7}
\\
P_{SR}(b,b)&=1-2a,
\label{e8}
\end{align}
with $b,a\in[0,1]$ solving the set of equations
\begin{align}
(1-2a)h(\tfrac1{2})
&=(1-a)h(b);
\\
(1-a)^2h(\bar{b}^2)
&=(1-a)h(b).
\label{e_bstar}
\end{align}
Hence
$
h(b)(2-h(b))=h(\bar{b}^2)
$, which has two roots in $(0,1)$, and we set $b^*\approx0.329454738503037$ as the larger root.
Let $a^*\approx 0.0788772927059232$ be the corresponding solution for $a$. These solutions led to 
\begin{align}
c^*\approx 0.3823455\label{e10}
\end{align}
which is the optimal value for the inequality in \eqref{e_6} and hence the best constant for the union-closed sets conjecture obtainable from Sawin's approach.
Define 
\begin{align}
f_{\alpha}(a,b):=\bar{\alpha}\mathbb{E}[h(\bar{S}\bar{T})]+\alpha\mathbb{E}[h(S\vee R\vee\min(S+R,1/2))]
-\mathbb{E}[h(S)].
\label{e_12}
\end{align}
Then from the equations 
\begin{align}
\partial_af_{\alpha}(a,b)
\,da+
\partial_bf_{\alpha}(a,b)
\,db&=0;
\\
d\mathbb{E}[S]=
\bar{b}\,da+\bar{a}\,db&=0,
\end{align}
we obtain
\begin{align}
\alpha^*=\frac{-\bar{a}[2\bar{a}h(\bar{b}^2)-h(b)]
+\bar{b}[2\bar{a}^2\bar{b}\log\frac{1-\bar{b}^2}{\bar{b}^2}+\bar{a}\log\frac{\bar{b}}{b}]}{-2\bar{a}[\bar{a}h(\bar{b}^2)-1]
+2\bar{a}^2\bar{b}^2\log\frac{1-\bar{b}^2}{\bar{b}^2}}.
\end{align}
Plugging in the values of $a^*$ and $b^*$ we obtain 
$\alpha^*\approx 0.0356069$.

{\bf Contribution.} In this paper we improve the previous best lower bound $c^*$ in \eqref{e10} for the union-closed sets conjecture by considering a new class of couplings of $X^n$ and $Y^n$ that ensures $P_{ST}$ is a mixture of i.i.d.\ distribution.
In other words, our coupling ensures the existence of some random variable $U$ such that $S$ and $T$ are i.i.d.\ conditioned on $U$.
The main observation is that the equality case of \eqref{e_6}, excluding the trivial case of $P_S$ supported on $\{0,1\}$, is given by \eqref{e7}-\eqref{e8}, which is a symmetric distribution but $[P_{ST}(s,t)]_{s,t\in\{b,1\}}$ is not positive semidefinite and hence not a mixture of i.i.d.\ distributions.

By taking a convex combination of the left side of \eqref{e_12} and $H(X_i\vee Y_i|X^{i-1},Y^{i-1})$ under the new coupling scheme, we can show analytically that $c^*$ can be strictly improved (Section~\ref{sec_analytical}). This argument is analogous to that of Sawin \cite{sawin2022improved} which proves strict improvement over the (unconditional) i.i.d.\ coupling by considering the limit of small convex combination weight, without numerical evaluation.
Note that $c^*$ was previously the best lower bound; 
in Remark~\ref{rem_mc} we show that if we apply a similar convex combination argument (in the regime of small combination weight) to Yu's general bound based on maximal correlation  \cite{yu2023dimension}, there is no improvement over $c^*$. 

In order to further numerically evaluate bounds obtainable from this new class of coupling, we then consider a convex combination of just Gilmer's original i.i.d.\ coupling and the conditionally i.i.d.\ coupling, without Sawin's max entropy coupling in the second term in \eqref{e_6}.
This allows us to show that considering two-mixture $P_{ST}$ is sufficient.
Under addition assumptions on the class of conditionally i.i.d.\ coupling employed, we further reduce the computation of the lower bound to a 9 dimensional optimization problem (Section~\ref{sec_numerical}).
For a basic instance from this class, numerical results using Matlab optimization package (interior point, sqp, active-set algorithms) with at least $10^5$ random initializations suggests that the best local optimizer has a certain structure.
Under this structural assumption, the global optimizer can be expressed as the solution to a set of analytic equations,
which in turns shows that the lower bound for the union-closed sets conjecture can be improved to approximately 0.382709 (Section~\ref{sec_numerical}).

\section{Conditionally IID Coupling}
In this section we explain the main idea for improvement based on conditionally i.i.d.\ coupling.
First, we formalize a method of sampling $X_i$ and $Y_i$ given $(X^{i-1},Y^{i-1})$ as follows:
\begin{defn}
We say $\Pi$ is a \emph{protocol} if for any 
$(s,t)\in[0,1]^2$, $\Pi_{s,t}$ is a distribution on $\{0,1\}^2$ satisfying: 
For any $(s,t)\in[0,1]^2$, we have $\mathbb{E}[X]=s$ and $\mathbb{E}[Y]=t$, 
where $(X,Y)\sim \Pi_{s,t}$.
\end{defn}
Given any $P_{X^n}$, a protocol $\Pi$ defines a method of randomly and sequentially generating $X^n$ and $Y^n$: for each $i=1,\dots,n$ and given $(X^{i-1},Y^{i-1})$, we generate
$X_i$ and $Y_i$ according to $\Pi_{s,t}$, where $s:=P_{X_i|X^{i-1}}(1|X^{i-1})$ and $t:=P_{X_i|X^{i-1}}(1|Y^{i-1})$, and $P_{X_i|X^{i-1}}$ is the conditional distribution induced by $P_{X^n}$.
This is a broad class of couplings that includes Gilmer's i.i.d.\ coupling \cite{gilmer2022constant}, Sawin's max-entropy coupling \cite{sawin2022improved}, Yu's  maximal correlation based coupling \cite{yu2023dimension}, and also the conditionally i.i.d.\ coupling that we will soon introduce.

The idea of taking convex combination of protocols in Proposition~\ref{prop2} can now be explained in larger generality:
Suppose that there are protocols $\Pi^{(1)}$, \dots, $\Pi^{(K)}$
such that for all distribution $\mu$ on $[0,1]$ with mean not exceeding $c$, we have 
\begin{align}
\sum_{k=1}^K w_k \inf_{P_{ST}\in \mathcal{C}_k(\mu)} \mathbb{E}[h(\Pi_{S,T}^{(k)}(0,0))]
\ge 
C\mathbb{E}[h(S)]
\label{e16}
\end{align}
where $w_1,\dots,w_K\ge 0$, $\sum_{k=1}^Kw_k=1$, $C>1$ are fixed, and $\mathcal{C}_k(\mu)$ is a set large enough to contain all the couplings of $\mu$ and $\mu$ that can possibly be induced by the $k$-protocol. 
For example, in Gilmer's approach, $K=1$, $\Pi_{s,t}^{(1)}={\rm Bern}(s)\times {\rm Bern}(t)$, and $\mathcal{C}_1$ is the singleton set containing $\mu\times \mu$.
Sawin's improvement can be interpreted as the case of $K=2$, and $\Pi_{s,t}^{(2)}$ is the most greedy coupling of ${\rm Bern}(s)$ and ${\rm Bern}(t)$ that maximizes the entropy of the OR function.
Correspondingly, $\mathcal{C}_2(\mu)$ is a rather large set -- probably no better choice than taking all symmetric couplings of $\mu$ and $\mu$.
In Yu's maximal correlation based formulation, $\Pi_{s,t}$ satisfies a maximal correlation upper bound, and so by tensorization $\mathcal{C}(\mu)$ is the set of distributions satisfying the same maximal correlation upper bound.

\begin{prop}\label{prop3}
If \eqref{e16} holds with constant $c>0$, then $c$ is a lower bound for the constant in the union-closed sets conjecture.
\end{prop}
\begin{proof}
We can identify the given family of subsets with binary sequences where a coordinate 1 indicates inclusion of the corresponding element.
Let $P_{X^n}$ be the equiprobable distribution on the resulting set of binary sequences. Let $(X^{(k)n},Y^{(k)n})$ be induced by protocol $\Pi^{(k)}$.
Assuming $P_{X_i}(1)<c$ for all $i$, then
analogous to \eqref{e_3}, we obtain from \eqref{e16} that
\begin{align}
\sum_{k=1}^K w_kH(X^{(k)n}\vee Y^{(k)n})
\ge CH(X^n).
\end{align}
Since the equiprobable distribution maximizes entropy among distributions supported on the a given set, this would imply that the given family is not union-closed.
\end{proof}

We next define a new class of $\Pi^{(3)}$ for which $h(\Pi_{s,t}^{(3)}(0,0))$ is larger than Gilmer's independent coupling, yet $\mathcal{C}_3(\mu)$ is strictly smaller than the set of symmetric couplings $\mathcal{C}_2$ as in Sawin's improvement.

\begin{defn}\label{def_ciid}
A protocol $\Pi$ is \emph{conditionally IID} if it can represented as 
\begin{align}
\Pi_{s,t}(x,y)
=\int Q_{u,s}(x)Q_{u,t}(y)P_U(du)
\label{e12}
\end{align}
where $Q_{u,s}$ is a Bernoulli distribution whose mean is a function of $(u,s)$, and $P_U$ is an arbitrary probability measure. 
By the isomorphism theorem of standard probability spaces, we can assume that $P_U$ is the uniform probability distribution on $[0,1]$.
\end{defn}

The max-entropy protocol of Sawin, which uses $\Pi^{(2)}_{s,t}(0,0)=1-s\vee t\vee\min(s+t,1/2)$, is not conditionally IID. Indeed, for $s=t<1/4$ we see $[\Pi_{s,t}^{(2)}(x,y)]_{x,y\in\{0,1\}}$ is matrix with diagonals equal 0, $1-2s$ and off-diagonals equal $s$, so it is not positive semidefinite, and hence cannot be the form \eqref{e12}.

\begin{example}
Let $\mathcal{X}$ and $\mathcal{Y}$ be two players with their local randomness, and 
let the common randomness be $U$ a random variable uniformly distributed on $[0,1]$.
Let $a(\cdot)$ be a measurable function on $[0,1]$.
Given $U=u$ and $s$ in $[0,1]$, let $X\sim Q_{u,s}$ be simulated this way: use the local randomness of $\mathcal{X}$ (which is not known to $\mathcal{Y}$) to simulate $B\sim {\rm Ber}(a(s))$. If $B=0$, then $X$ is generated as a ${\rm Ber}(s)$ random variable using the local randomness of $\mathcal{X}$; if $B=1$, we set $X=1_{U<s}$. 
Then $X\sim {\rm Ber}(s)$ (conditioned on $s$). 
Generate $Y$ similarly using $u$ and $t$. 
We have
\begin{align}
&\quad \mathbb{P}[X=Y=0|U]
\nonumber\\
&=
a(s)a(t)1_{U\ge s\vee t}
+a(s)(1-a(t))1_{U\ge s}\bar{t}
\nonumber\\
&\quad+(1-a(s))a(t)\bar{s}1_{U\ge t}
+(1-a(s))(1-a(t))\bar{s}\bar{t}
\end{align}
and hence
\begin{align}
\Pi_{s,t}(0,0)
&=a(s)a(t)(1-s\vee t)
+(1-a(s))a(t)\bar{s}\bar{t}
\nonumber\\
&\quad+a(s)(1-a(t))\bar{s}\bar{t}
+(1-a(s))(1-a(t))\bar{s}\bar{t}
\\
&=\bar{s}\bar{t}+a(s)a(t)(\bar{s}\wedge \bar{t}-\bar{s}\bar{t}).
\end{align}
\end{example}
If we choose $a()$ to maximize $h(\Pi_{s,t}(0,0))=h(\bar{t}^2+a(t)^2t\bar{t})$ for any given $t$, we are led to
\begin{align}
a(t)=
\left\{
\begin{array}{cc}
 0    &  t\le 1-\frac1{\sqrt{2}}\\
  \sqrt{\frac{1-2\bar{t}^2}{2t\bar{t}}}   & 1-\frac1{\sqrt{2}}<t\le \frac1{2}
  \\
  1 & t>\frac1{2}
\end{array}
\right..
\label{e_23}
\end{align}

\begin{example}\label{ex2}
Let $P_U$ be the uniform distribution on $[0,1]$, and let $f()$ be a measurable function on $[0,1]$ satisfying $0\le f(\bar{s})\le s\wedge\bar{s}$. 
Given $u,s\in[0,1]$, let $X\sim Q_{u,s}$ be simulated this way:
set $X=0$ with probability $\bar{s}+f(\bar{s})(1_{u>\frac1{2}}-1_{u\le \frac1{2}})$, and $X=1$ with the remaining probability.
Then 
$
\Pi_{s,t}(0,0)=\bar{s}\bar{t}+f(\bar{s})f(\bar{t}).
$
\end{example}

The coupling induced by a conditional IID protocol can be interpreted as simulating each $X_i$ by looking only at $X^{i-1}$ and $U^{i-1}$, without referencing the $Y$-sequence, and similarly for each $Y_i$.
Therefore, $X^i$ and $Y^i$ are conditionally i.i.d.\ given $U^i$ for all $n$.
Hence $\mathbb{E}[X_i|X^{i-1}]$ and $\mathbb{E}[Y_i|Y^{i-1}]$, which are functions of $X^{i-1}$ and $Y^{i-1}$, are conditionally i.i.d.
We can therefore take
\begin{align}
\mathcal{C}_3(\mu)
:=\{\textrm{couplings of $\mu$ and $\mu$}
\}\cap {\rm cl}\,{\rm conv}\{\textrm{symmetric rank-1 measures}\}
\end{align}
where ${\rm conv}$ denotes the convex hull, and the closure ${\rm cl}$ is with respect to weak convergence (ensuring that we can invoke Krein-Milman later).

\section{Analytic Proof of Strict Improvement}\label{sec_analytical}
Recall that $c^*\approx0.3823455$, defined around \eqref{e_bstar} as solution to analytic equations, was the previous best lower bound for the union-closed sets conjecture.
The main result of the section is the following:
\begin{thm}\label{thm6}
The union-closed sets conjecture holds with a constant strictly larger than $c^*$.
\end{thm}
The idea is to use \eqref{e16} and perturb the previous best scheme with a conditionally IID protocol.
First, we prove the following result which allows us to not worry about the case of $P_S$ supported on $\{0,1\}$, a trivial equality case for \eqref{e16}.
This observation was briefly mentioned (without quantitative bound) in Sawin's paper \cite{sawin2022improved}.
\begin{lem}\label{lem3}
Let $c\in(0,1)$.
Suppose that $\mu_n$ is a sequence of probability measures with mean equal to $c$ and converging weakly to $\mu^*$, where $\mu^*$ is the (unique) probability measure supported on $\{0,1\}$ and with mean equal to $c$. 
Then 
\begin{align}
\liminf_{n\to\infty}
\frac{\int h(\bar{s}\bar{t})\mu_n(ds)\mu_n(dt)}{\int h(s)\mu_n(ds)}
\ge 2\bar{c}.
\end{align}
\end{lem}
\begin{proof}
By the assumption of weak convergence, we can pick a sequence $\epsilon_n\downarrow 0$ such that 
\begin{align}
\mu_n([\epsilon_n,1-\epsilon_n]) \le \epsilon_n.
\end{align}
Let $\mathcal{A}_n:=[0,\epsilon_n)$ and $\mathcal{B}_n:=[\epsilon_n,1]$.
Using $h(x)\sim x\ln\frac1{x}$ as $x\to0$, we see
\begin{align}
\lim_{n\to\infty}\inf_{s\in\mathcal{A}_n,t\in\mathcal{B}_n}\frac{h(\bar{s}\bar{t})}{h(\bar{t})}=1.
\end{align}
Therefore we have 
\begin{align}
\liminf_{n\to\infty}\frac{\int_{\mathcal{A}_n\times \mathcal{B
}_n} h(\bar{s}\bar{t})\mu_n(ds)\mu_n(dt)}
{\int_{\mathcal{B}_n}{h(s)\mu_n(ds)}}
\ge \lim_{n\to\infty}\mu_n(\mathcal{A}_n)=\bar{c}.
\label{e5}
\end{align}
Moreover, applying Proposition~\ref{prop1} to the probability measure (conditional probability) $\frac1{\mu_n(\mathcal{A}_n)}\mu_n|_{\mathcal{A}_n}$, whose mean is smaller than $\epsilon_n$, we obtain 
\begin{align}
\liminf_{n\to\infty}\frac{\int_{\mathcal{A}_n\times \mathcal{A}_n}h(\bar{s}\bar{t})\mu_n(ds)\mu_n(dt)}{\int_{\mathcal{A}_n}h(s)\mu_n(dt)}
\ge 2\lim_{n\to\infty}\mu_n(\mathcal{A}_n)
=2\bar{c}.
\label{e6}
\end{align}
The lemma then follows by \eqref{e5}-\eqref{e6} and 
$\int_{\mathcal{A}_n\times \mathcal{B
}_n} h(\bar{s}\bar{t})\mu_n(ds)\mu_n(dt)=\int_{\mathcal{B}_n\times \mathcal{A
}_n} h(\bar{s}\bar{t})\mu_n(ds)\mu_n(dt)$.
\end{proof}
Theorem~\ref{thm6} now follows by Proposition~\ref{prop3} and the following observation:
\begin{lem}\label{lem8}
Let $a(\cdot)$ be as in \eqref{e_23} and let $\Pi^{(3)}$ be the corresponding conditionally i.i.d.\ protocol.
There exists $\beta\in(0,1)$, $c'>c^*$, and $C>1$ such that 
\begin{align}
\bar{\alpha}^*\bar{\beta}\mathbb{E}[h(\bar{S}\bar{T})]+\alpha^*\bar{\beta}\mathbb{E}[h(S\vee R_1\vee\min(S+R_1,1/2))]
\nonumber\\
+\beta\mathbb{E}[h(\bar{S}\bar{R}_2+a(S)a(R_2)(\bar{S}\wedge \bar{R}_2-\bar{S}\bar{R}_2))]
\ge C\mathbb{E}[h(S)]
\label{e23}
\end{align}
whenever $\mathbb{E}[S]\le c'$,
$P_{ST}\in \mathcal{C}_1(P_S)$
$P_{SR_1}\in \mathcal{C}_2(P_S)$, $P_{SR_2}\in\mathcal{C}_3(P_S)$.
\end{lem}
\begin{proof}
Let $P_S^*$ be the two-point distribution where $\mathbb{P}[S=b^*]=1-a^*$ and $\mathbb{P}[S=1]=a^*$, with $a^*$ and $b^*$ defined around \eqref{e10}. 
Note that
\begin{align}
\inf_{P_{SR_2}\in\mathcal{C}_3(P_S^*)}\frac{\mathbb{E}[h(\Pi_{S,R_2}^{(3)}(0,0))]
}
{\mathbb{E}[h(S)]}
>1.
\label{e24}
\end{align}
Indeed, if the infimum on the left side of \eqref{e24} is over $P_{SR_2}\in\mathcal{C}_2(P_S^*)$ instead, then the infimum equals 1, achieved at $P_{SR_2}^*$ defined as 
\begin{align}
P_{SR_2}^*(1,1)&=0;
\\
P_{SR_2}^*(b^*,1)&=
P_{SR_2}^*(1,b^*)=a^*;
\\
P_{SR_2}^*(b^*,b^*)&=1-2a^*.
\end{align}
Since $P_{SR_2}^*\in P_{SR_2}$ belongs to $\mathcal{C}_2(P_S^*)$ but not $\mathcal{C}_3(P_S^*)$, and $\mathcal{C}_3(P_S^*)\subseteq \mathcal{C}_2(P_S^*)$, we see \eqref{e24} holds.

Next, we claim that there exists some $\delta>0$ such that 
\begin{align}
\inf_{P_S:\,W_1(P_S,P_S^*)<\delta}\,\inf_{P_{SR_2}\in\mathcal{C}_3(P_S)}\frac{\mathbb{E}[h(\Pi_{S,R_2}^{(3)}(0,0))]
}
{\mathbb{E}[h(S)]}
>1,
\label{e24}
\end{align}
where $W_1$ denotes the Wasserstein-1 distance.
Indeed, for any such $P_S$ in \eqref{e24} and for $ P_{SR_2}\in\mathcal{C}_3(P_S)$, we can construct a $P_{SR_2}^{**}\in\mathcal{C}_3(P_S^*)$ such that $W_1(P_{SR_2},P_{SR_2}^{**})\le 2\delta$,
by using the optimal transport (stochastic) map in the definition of the Wasserstein distance.
By choosing $\delta>0$ small enough we can also assume that $P_S$ is bounded away from measures supported on $\{0,1\}$, so that the denominator in \eqref{e24} is bounded away from 0.
Then the claim follows by the continuity of the numerator and denominator in \eqref{e24} with respect to weak convergence (equivalently, with respect to $W_1$).

Now define 
\begin{align}
g(P_S):=\inf_{P_{SR_2}}\frac{\bar{\alpha}^*\mathbb{E}[h(\Pi_{S,T}^{(1)}(0,0))]+\alpha^*\mathbb{E}[h(\Pi_{S,R_1}^{(2)}(0,0))]}
{\mathbb{E}[h(S)]}
\end{align}
where $T$ is an i.i.d.\ copy of $S$.
Set
\begin{align}
C_{\delta}:=\inf_{P_S:\,W_1(P_S,\,P_S^*)\ge \delta,\,\mathbb{E}[S]\le c^*}
g(P_S).
\end{align}
Then $C_{\delta}>1$ for any $\delta>0$. 
Indeed, suppose that $P_S^n$ is a weakly convergent sequence satisfying $W_1(P_S^n,P_S^*)\ge\delta$, 
$\mathbb{E}_{P_S^n}[S]\le c^*$, and $\lim_{n\to\infty}g(P_S^n)=C_{\delta}$.  
If $P_S':=\lim_{n\to\infty}P_S^n$ is supported on $\{0,1\}$, then by Lemma~\ref{lem3} we have $C_{\delta}\ge 2\bar{c}^*\bar{\alpha}^*>1$;
otherwise, $C_{\delta}=g(P_S')$ and $\mathbb{E}_{P_S'}[h(S)]>0$,
and we conclude from the condition of strict inequality in Proposition~\ref{prop2} that $C_{\delta}>1$. 

Now we can set $\beta\in(0,1)$ as any number satisfying $\bar{\beta}C_{\delta}>1$.
Since Proposition~\ref{prop2} showed
\begin{align}
\inf_{P_S:\,\mathbb{E}[S]\le c^*}g(P_S)\ge 1,
\end{align}
together with \eqref{e24} we have established 
\begin{align}
\inf_{P_S:\, W_1(P_S,P_S^*)<\delta,\,
\mathbb{E}[S]\le c^*}f(\beta,P_S)> 1
\end{align}
where 
\begin{align}
&f(\beta,P_S):=
\nonumber\\
&\inf_{P_{SR_1}\in\mathcal{C}_2(P_S),P_{SR_2}\in\mathcal{C}_3(P_S)}\tfrac{\bar{\alpha}^*\bar{\beta}\mathbb{E}[h(\Pi_{S,T}^{(1)}(0,0))]+\alpha^*\bar{\beta}\mathbb{E}[h(\Pi_{S,R_1}^{(2)}(0,0))]
+\beta\mathbb{E}[h(\Pi_{S,R_2}^{(3)}(0,0))]}
{\mathbb{E}[h(S)]}.
\label{e30}
\end{align}
Combining \eqref{e30} with $\bar{\beta}C_{\delta}>1$, we actually have 
\begin{align}
\inf_{P_S:\,
\mathbb{E}[S]\le c^*}f(\beta,P_S)> 1.
\label{e31}
\end{align}

Finally, suppose that $P_S^n$ is a weakly convergent sequence that such that $\mathbb{E}_{P_S^n}[S]=c^*+\frac1{n}$ and 
\begin{align}
\lim_{n\to\infty}f(\beta,P_S^n)
\le 
\lim_{n\to\infty}\inf_{P_S:\,
\mathbb{E}[S]\le c^*+\frac1{n}}f(\beta,P_S).
\label{e32}
\end{align}
If $P_S^n$ converges to a probability measure supported on $\{0,1\}$, we conclude from Lemma~\ref{lem3} that 
$\lim_{n\to\infty}f(\beta,P_S^n)\ge 2\bar{c}^*$;
otherwise, $\lim_{n\to\infty}f(\beta,P_S^n)=f(\beta,\lim_{n\to\infty}P_S^n)>1$ by \eqref{e31}. 
We therefore established that the right side of \eqref{e32} is strictly larger than 1, which is the claim of the lemma.
\end{proof}
\begin{rem}\label{rem_mc}
Yu \cite{yu2023dimension} considered $\Pi$ that ensures the maximal correlation coefficient of the binary pair distribution $\Pi_{s,t}$ is upper bounded by a given $\rho\in[0,1]$ for any $(s,t)$.
Then by a tensorization property, it follows that $\mathcal{C}$ can be taken to be the set of distributions on $[0,1]^2$ with maximal correlation upper bounded by $\rho$.
The cases of $\rho=0$ and $\rho=1$ reduces to the i.i.d.\ coupling and the max-entropy coupling, respectively.
A natural question is whether we can pick some $\rho\in(0,1)$ and apply a similar argument as Lemma~\ref{lem8} to show strict improvement on $c^*$.
The answer is positive only if for some $\rho\in(0,1)$,
\begin{align}
P_{ST}(b^*,b^*)h(\Pi_{b^*b^*}(0,0))
-(1-a^*)h(b^*)>0,\label{e_44}
\end{align}
where $P_{ST}$ is the coupling of $P_S=P_T$ defined by $P_S(b^*)=1-a^*$ and $P_S(1)=a^*$ with maximal correlation upper-bounded by $\rho$ such that $P_{ST}(b^*,b^*)$ is minimized, and 
$\Pi_{b^*b^*}$ is the coupling of two Bernoulli $b^*$ distributions with maximal correlation upper-bounded by $\rho$ such that $h(\Pi_{b^*b^*}(0,0))$ is maximized.
Since the maximal correlation of discrete distributions can be computed as the second singular value of a matrix,
we have $\Pi_{b^*b^*}(0,0)=\min\{\bar{b}^{*2}+b^*\bar{b}^*\rho,0.5\}$ and $P_{ST}(b^*,b^*)=\max\{\bar{a}^{*2}-a^*\bar{a}^*\rho,1-2\bar{a}^*\}$ \cite{yu2023dimension}.
Then we can verify that the left side of \eqref{e_44} is, in fact, negative for all $\rho\in(0,1)$.
\end{rem}

\section{Cardinality Reduction}\label{sec_reduction}
In this section we remove the max-entropy coupling term in \eqref{e23}, and simplify \eqref{e42} to a finite-dimension optimization under certain assumptions (Theorem~\ref{thm12} ahead), which enables numerical evaluation of the bound.
As before, let $\mathcal{C}_1$ be the set of measures $P_{SR}$ on $[0,1]^2$ under which $S$ and $R$ are i.i.d., and let $\mathcal{C}_3$ be the closure of the convex hull of $\mathcal{C}_1$.
We focus on the following optimization:
\begin{align}
\inf_{P_{SR_2}\in\mathcal{C}_3,\,\mathbb{E}[S]\le c}
\left\{\bar{\beta}\mathbb{E}[h(\bar{S}\bar{T})]
+\beta\mathbb{E}[h(\Pi_{S,R_2}(0,0))]-
\mathbb{E}[h(S)]
\right\}
\label{e42}
\end{align}
where $c>0$, and $S$ and $T$ are i.i.d. 

\begin{thm}\label{thm7}
For any conditionally IID protocol $\Pi$ as in Definition~\ref{def_ciid},
the infimum in \eqref{e42} is achieved by some $P_{SR_2}$ which is a mixture of two i.i.d.\ distributions, i.e.,
\begin{align}
P_{SR_2}(s,r)=\mathbb{E}[P_{S|W}(s|W)P_{S|W}(r|W)] 
\label{e43}
\end{align}
for some binary random variable $W$ and conditional distribution $P_{S|W}$.
\end{thm}
\begin{proof}
The fact that the infimum is achievable follows from the weak compactness of measures (Prokhorov theorem).
For any $c'>0$, the map from $\mathcal{C}_3\cap\{P_{SR_2}\colon \mathbb{E}[S]=c'\}$ to $\mathbb{E}[h(\bar{S}\bar{T})]$ is concave, as shown in \cite{alweiss2022improved}, 
so is its composition with  the linear map from $P_{SR_2}$ to $P_S$.
The last two summands in \eqref{e42} are linear in $P_{SR_2}$. 
Therefore \eqref{e42} is achieved at the extreme points of $\mathcal{C}_3\cap\{P_{SR_2}\colon \mathbb{E}[S]=c'\}$ for some $c'>0$.

It remains to show that the extreme points are mixtures of two i.i.d.\ distributions, from which the theorem follows by Krein-Milman.
For any $P_{SR_2}\in\mathcal{C}_3$
we can write $P_{SR_2}$ by \eqref{e43} for some (not necessarily binary) random variable $W$.
We can assume that $W$ has finite support, as the general case will then follow by a limiting argument.
Then $\mathbb{E}[S|W]$ is a random variable on $\mathbb{R}$ with mean equal to $c'$.
The extreme points in the set of probability measures on $\mathbb{R}$ with mean equal to $c'$ are mixtures of two delta measures with mean equal to $c'$,  
Therefore we can express the distribution of $\mathbb{E}[S|W]$ as a convex combination of mixture of two delta measures, and hence express $P_{SR_2}$ as a convex combination of mixtures of two i.i.d.\ measures with mean equal to $c'$.
This establishes the claim about extreme points and hence the theorem statement.
\end{proof}

Theorem~\ref{thm7} implies that it is sufficient to consider
\begin{align}
P_{SR_2}(s,r)=qP_1(s)P_1(r)
+\bar{q}P_0(s)P_0(r)
\label{e44}
\end{align}
where $q\in[0,1]$ and $P_0$ and $P_1$ are probability measures on $[0,1]$.
Further simplification is possible in some settings.
The idea is to use convexity of certain functionals, which in turn relies on the positive semidefiniteness of certain quadratic forms.
First, we observe the following about some matrices generated by polynomials.

\begin{lem}\label{lem10}
Let $p(\cdot)$ and $p_1(\cdot,\cdot)$ be given polynomials. For positive integer $k$, let $B_k$ be the matrix where the $(i,j)$-th entry equals the coefficient of the monomial $x^iy^j$ in the expansion of $p_1(x,y)p^k(x)p^k(y)$.
Then there exists $C>0$ such that the operator norm of $B_k$ is upper bounded by $C^k$ for all $k$.
\end{lem}
\begin{proof}
Suppose that $p(x)=\sum_{m=0}^Da_mx^m$,
and $p^k(x)=\sum_{m=0}^{Dk}a^{(k)}_mx^m$.
Let us focus on the case of $p_1(\cdot,\cdot)=1$, as the case of monomial $p_1$ will then follow with exactly the same spectral norm, and then the general $p_1$ case will follow by subadditivity. 
Then $B_k$ is a symmetric rank-one matrix, 
and from the large deviation analysis,
the square of its operator norm is 
\begin{align}
\sum_{m=1}^{Dk}(a^{(k)}_m)^2
=
\exp\left(2k\sup_{P_M}\left\{\mathbb{E}[\log(a_M)]+H(P_M)\right\}+o(k)\right)
\end{align}
where the supremum is over $P_M$ a distribution on $\{0,1,\dots, D\}$.
Therefore the operator norm grows at most exponentially in $k$.
\end{proof}
Note that if $p_1(x,y)p^k(x)p^k(y)$ is a symmetric polynomial whose max degree in $x$ is $L$, then Lemma~\ref{lem10} implies that
\begin{align}
\int p_1(x,y)p^k(x)p^k(y)\gamma(x)\gamma(y)dxdy
\le C^k(\sum_{i=0}^L\int x^i \gamma(x))^2
\end{align}
for any $\gamma(\cdot)$.
We can use this fact to establish the following:
\begin{lem}\label{lem9}
Fix $p(\cdot)$ a polynomial satisfying $p(1)=0$.
Consider $\Pi$ in Example~\ref{ex2}, with $f(x)=lxp(x)$.
For sufficiently small $l>0$, the following holds:
For any $c,d\in[0,1]$, the map $\mu\mapsto \int h(\Pi_{s,r}(0,0))\mu(ds)\mu(dt)$ restricted to the set of probability measures satisfying
\begin{align}
\mathbb{E}_{\mu}(S)&=c;
\label{e49}
\\
\mathbb{E}_{\mu}[f(\bar{S})]&=d,
\label{e50}
\end{align}
is concave.
\end{lem}
\begin{proof}
Note that for sufficiently small $l$, $0\le f(\bar{s})\le s\wedge\bar{s}$ is satisfied and so $\Pi$ is a well-defined protocol.
For notation simplicity, we write $x:=\bar{s}$ and $y:=\bar{t}$, and $I:=\Pi_{s,r}(0,0)=xy+f(x)f(y)$.
The goal can be rephrased as showing $ -\int h(I)\mu(dx)\mu(dy)\ge 0$ for measures $\mu$ on $[0,1]$ satisfying
$\int d\mu=0$, $\int x\mu(dx)=0$ and $\int f(x)\mu(dx)=0$.

Similar to \cite{alweiss2022improved}, we apply integration by parts twice
to obtain
\begin{align}
-\int h(I)\mu(dx)\mu(dy)
&=-\int \partial_x\partial_yh(I)\gamma(x)\gamma(y)dxdy
\\
&=\int\left(
I_{xy}\log\frac{I}{1-I}
+\frac{I_xI_y}{I(1-I)}
\right)
\gamma(x)\gamma(y)dxdy
\label{e48}
\end{align}
where $\gamma(x):=\mu([0,x])$ and $I_x$ denotes the derivative of $I$ in $x$. 
We then analyze the terms in \eqref{e48} separately to show the nonnegativity of \eqref{e48}.

First, 
\begin{align}
A_1&:=\int I_{xy}(\log I)\gamma(x)\gamma(y)
\\
&=\int(1+f'(x)f'(y))\log(xy+f(x)f(y))\gamma(x)\gamma(y)
\\
&=\int \log(xy)\gamma(x)\gamma(y)
+\int f'(x)f'(y)\log(xy)\gamma(x)\gamma(y)
\nonumber\\
&\quad+\int(1+f'(x)f'(y))\log\left(1+\frac{f(x)f(y)}{xy}\right)\gamma(x)\gamma(y).
\label{e51}
\end{align}
Now $\int\log (y)\gamma(x)\gamma(y)=-\int \log (y)\gamma(y)\int d\mu=0$ by integration by parts.
Similarly, $\int f'(x)f'(y)\log(y)\gamma(x)\gamma(y)=-\int f'(y)\log(y)\gamma(y)\int f(x)\mu(dx)=0$.
The third term in \eqref{e51} can be Taylor expanded as 
\begin{align}
\int (1+l^2(xp(x))'(yp(y))')
\sum_{k=1}^{\infty}\frac{(-1)^{k+1}}{k}
l^{2k}p^k(x)p^k(y)\gamma(x)\gamma(y),
\end{align}
which, by Lemma~\ref{lem10}, is lower bounded by 
\begin{align}
&\quad-\sum_{k=1}^{\infty}
l^{2k}\sum_{m=0}^{(k+1)D}C^k\left(\int x^m\gamma(x)\right)^2
\nonumber\\
&\ge -\sum_{m=0}^{\infty}\frac{l^2C}{1-l^2C}\left(\int x^m\gamma(x)\right)^2
\\
&\ge -2Cl^2\sum_{m=0}^{\infty}\left(\int x^m\gamma(x)\right)^2
\label{e_54}
\end{align}
for sufficiently small $l>0$ (by which we mean $l$ is smaller than some positive threshold depending on $p(\cdot)$),
where $D$ denotes the degree of $p(\cdot)$, and $C>0$ depends only on $p(\cdot)$.

Second, 
\begin{align}
A_2&:=\int I_{xy}\log\frac1{1-I}\gamma(x)\gamma(y)
\\
&=\int (1+l^2(xp(x))'(yp(y))')\sum_{k=1}^{\infty}\frac1{k}(xy+f(x)f(y))^k\gamma(x)\gamma(y)
\label{e54}
\\
&=\sum_{k=1}^{\infty}\int F_k(x)F_k(y)\gamma(x)\gamma(y)\ge 0
\end{align}
where $F_k(\cdot)$ are certain polynomials arising from applying binomial expansion to \eqref{e54}.

Third,
\begin{align}
A_3&:=\int\frac{I_xI_y}{I(1-I)}
\gamma(x)\gamma(y)
\nonumber\\
&=\int\frac{I_xI_y}{I}\gamma(x)\gamma(y)
+\int I_xI_y\sum_{k=0}^{\infty}I^k\gamma(x)\gamma(y).
\label{e57}
\end{align}
Denote by $A_{31}$ and $A_{32}$ the two integrals in \eqref{e57},
and define $q(x):=(xp(x))'$. 
We have
\begin{align}
A_{31}&=
\int\frac{(1+l^2q(x)p(y))(1+l^2p(x)q(y))}{1+l^2p(x)p(y)}\gamma(x)\gamma(y)
\\
&=\int \gamma(x)\gamma(y)
\nonumber\\
&+\int (l^2q(x)p(y)+l^2p(x)q(y)+l^4q(x)q(y)p(x)p(y))\gamma(x)\gamma(y)
\nonumber\\
&+\int(1+l^2q(x)p(y))(1+l^2p(x)q(y))\sum_{k=1}^{\infty}(-1)^kl^{2k}p^k(x)p^k(y)\gamma(x)\gamma(y)
\label{e59}
\end{align}
Denote the 3 integrals in \eqref{e59} by $A_{311}$, $A_{312}$ and $A_{313}$, respectively.
We have
\begin{align}
A_{312}\ge -l^2C\sum_{m=0}^{2D+1}(\int x^m\gamma(x))^2
\end{align}
and similarly to \eqref{e54},
\begin{align}
A_{313}\ge -2l^2C\sum_{m=0}^{\infty}\left(\int x^m\gamma(x)\right)^2
\end{align}
for some $C>0$ depending on $p(\cdot)$.
Next,
\begin{align}
A_{32}&=\int xy(1+l^2q(x)p(y))(1+l^2p(x)q(y))
\nonumber\\
&\cdot\sum_{k=0}^{\infty}x^ky^k(1+l^2p(x)p(y))^k\gamma(x)\gamma(y)
\\
&=\int \sum_{k=0}^{\infty}x^{k+1}y^{k+1}(1+l^2p(x)p(y))^k\gamma(x)\gamma(y)
\nonumber\\
&+\int (l^2p_1(x,y)+l^4p_2(x,y))
\sum_{k=0}^{\infty}x^{k+1}y^{k+1}(1+l^2p(x)p(y))^k\gamma(x)\gamma(y)
\label{e64}
\end{align}
where $p_1$ and $p_2$ are symmetric polynomials whose definitions can be seen from the expansion of terms.
The degree in $x$ of $p_1$ and $p_2$ are $D$ and $2D$ respectively.
Let $A_{321}$ and $A_{322}$ be the two integrals in \eqref{e64}.
There exists $a_0,\dots,a_{2D}>0$ such that 
\begin{align}
-\int(p_1(x,y)+l^2p_2(x,y))\gamma(x)\gamma(y)
\le \sum_{m=0}^{2D}a_m(\int x^m\gamma(x))^2
\end{align}
whenever $l<1$.
Recall the Schur product theorem about the positive semidefiniteness of the elementwise product of PSD matrices; it follows that 
\begin{align}
-A_{322}\le
l^2\int \sum_{m=0}^{2D}a_mx^my^m
\sum_{k=0}^{\infty}x^{k+1}y^{k+1}(1+l^2p(x)p(y))^k\gamma(x)\gamma(y).
\label{e66}
\end{align}
We can expand $(1+l^2p(x)p(y))^k$ in \eqref{e66};
for each $m\in\{0,\dots,2D\}$, $k\ge 0$, and $n\in\{0,\dots,k\}$,
the coefficient for $\int x^{m+k+1}y^{m+k+1}p^n(x)p^n(y)\gamma(x)\gamma(y)$ in \eqref{e66} is 
$
l^{2n+2}a_m{{k}\choose{n}}
$.
On the other hand, the coefficient for the same term in $A_{321}$ is $l^{2n}{{m+k}\choose{n}}\ge l^{2n}{k\choose{n}}$.
This shows that 
$
-A_{322}\le l^2\max\{a_1,\dots,a_{2D}\}A_{321}
$ and hence 
\begin{align}
A_{32}\ge \frac1{2} A_{321}
\ge
\frac1{2}\int \sum_{k=0}^{\infty}x^{k+1}y^{k+1}\gamma(x)\gamma(y)
\end{align}
for $l$ sufficiently small.
Returning to \eqref{e59}, we have
\begin{align}
A_3&=A_{31}+A_{32}
\\
&\ge \frac1{2}(1-C'l^2 )\sum_{k=0}^{\infty}(\int x^k\gamma(x))^2
\\
&\ge\frac1{4}\sum_{k=0}^{\infty}(\int x^k\gamma(x))^2
\end{align}
for $l$ sufficiently small and where $C'>0$ is some constant.
And so $A_1+A_2+A_3\ge \frac1{8}\sum_{k=0}^{\infty}(\int x^k\gamma(x))^2$ for $l$ sufficiently small, establishing the desired nonnegativity of \eqref{e48}.
\end{proof}

Now we can state the main result of this section:
\begin{thm}\label{thm12}
Fix $p(\cdot)$ a polynomial satisfying $p(1)=0$.
Consider $\Pi$ in Example~\ref{ex2}, with $f(x)=lxp(x)$.
Let $c,\beta,q$ be fixed, and consider the optimization of \eqref{e42} over $P_0$ and $P_1$ in \eqref{e44}.
For sufficiently small $l>0$, the optimal value can be achieved by $(P_0,P_1)$ of the form
\begin{align}
P_0&=a_1\delta_{b_0}
+a_2\delta_{b_2}
+a_3\delta_{b_4};
\label{e_76}
\\
P_1&=a_1\delta_{b_1}
+a_2\delta_{b_3}
+a_3\delta_{b_5},
\label{e_77}
\end{align}
where $b_0,\dots,b_5\in[0,1]$, and $(a_1,a_2,a_3)$ is on the probability simplex.
\end{thm}
\begin{proof}
Pick arbitrary $c'\in[0,c]$ and $d\in[0,1]$, and  
consider the map from $(P_0,P_1)$ to the objective function in \eqref{e42}, restricted to the $(P_0,P_1)$ satisfying
\begin{align}
q\int xP_1(dx)+\bar{q}\int xP_0(dx)&=1-c';
\label{e75}
\\
q\int f(x) P_1(dx)+\bar{q}\int f(x)P_0(dx)&=d.
\label{e76}
\end{align}
From Lemma~\ref{lem9},
we see that this is a
is a concave functional on a convex set, hence the infimum is achieved at extremal points (Krein-Milman theorem).
For the set of $(P_0,P_1)$ where $P_0$ and $P_1$ are from the set of probability measures, the extremal points are of the form $(\delta_{b_0},\delta_{b_1})$.
If we add the two linear constraints in \eqref{e75}-\eqref{e76}, the extremal points in this restricted convex set are convex combinations of three such delta measure pairs
(using an argument similar to the proof of \eqref{thm7}), which is the claim of the theorem.

An alternative argument for the last part is as follows: for any $(P_0,P_1)=\sum_{k=1}^Ka_k(\delta_{b_k},\delta_{b'_k})$, $a_k>0$, $\sum_{k}^Ka_k=1$ with $K\ge 4$, we consider variations of $a_1,\dots,a_4$ while retaining $\sum_{k=1}^4 a_k$, $\sum_{k=1}^4 a_k(\bar{q}b_k+qb'_k)$, and the values of $a_5,\dots,a_K$. This yields a 2-dimensional polygon. Lemma~\ref{lem9} implies that the quadratic form $\int h(\Pi_{s,r}(0,0))\mu(ds)\mu(dt)$ restricted to $\int \mu=0$ and $\int s\mu(ds)=0$ has positive signature at most 1, so on this 2-dimensional polygon there must be a line along which we can move $q_1,\dots,q_4$ so that $\int h(\Pi_{s,r}(0,0))\mu(ds)\mu(dt)$ is a quadratic function with nonpositive leading coefficient, where we set $\mu=\bar{q}P_0+qP_1$.
Thus we can move along this line until hitting the boundary of the polygon without increasing the objective value, hence showing $K$ can be reduced if $K\ge 4$. 
For $(P_0,P_1)$ not finitely supported, we can use approximation argument and the fact that the weak limit of $K$-supported distributions is also $K$ supported; see similar argument in \cite[Lemma~6]{alweiss2022improved}.
\end{proof}

\section{Numerical Evaluation}\label{sec_numerical}
In this section we focus on a basic instance of Example~\ref{ex2}, $f(x)=x\bar{x}$, numerically evaluate the largest $c$ for \eqref{e42} to be nonnegative,
and discuss its implication for the union-closed sets conjecture (Theorem~\ref{thm13}).

\subsection{Positive-Semidefiniteness}\label{sec_psd}
Recall that Theorem~\ref{thm12} reduces the optimization to a 9-dimensional one involving $(a_1,a_2,q,b_0,b_2,b_4,b_1,b_3,b_5)$, where $q$ denotes the weight for the $P_1$ component of the mixture. 
Theorem~\ref{thm12} 
is based on Lemma~\ref{lem9}, which states that the quadratic form 
$\mu\mapsto -\int h(\Pi_{s,r}(0,0))\mu(ds)\mu(dt)$ is positive semidefinite in the codimension-2 subspace specified by \eqref{e49}-\eqref{e50}, if $f(x)=lxp(x)$, $l$ is sufficiently small, and $p$ is a polynomial with $p(1)=0$.
For $l=1$, we verify the positive semidefiniteness of this quadratic form by numerically computing the eigenvalues of the matrix $[-h(\Pi_{s,r}(0,0))]_{s,t\in\mathcal{G}}$ on the codimension-2 subspace, where $\mathcal{G}$ is a grid on $[0,1]$ with separation 0.0004.
Our Matlab code can be found in \url{https://jingbol.web.illinois.edu/frankl3.m}.
The min eigenvalue is $-2.3685\times 10^{-14}$, which is negligible considering the numerical errors for computation of large matrices.
For comparison, the numerical precision of Matlab is $2.22\times 10^{-16}$; for the case of i.i.d.\ coupling where $\Pi_{s,t}(0.0)=\bar{s}\bar{t}$, the positive-semidefiniteness was rigorously shown in \cite{alweiss2022improved}, and the numerically evaluated minimum eigenvalue is 
$-2.4206\times10^{-14}$. 

As another approach of verifying positive semidefiniteness, 
consider $x=\bar{s}$,
$y=\bar{t}$ and $I:=\Pi_{s,t}(0,0)=xy(1+\bar{x}\bar{y})$.
For measures $\mu$ on $[0,1]$ satisfying $\int \mu(dx)=0$, $\int x\mu(dx)=0$, $\int x^2\mu(dx)=0$, we have 
\begin{align}
-\int h(I)\mu(dx)\mu(dy)
&=\int xy(1+\bar{x}\bar{y})\log\left(1-\frac{x+y-xy}{2}\right)\mu(dx)\mu(dy)
\nonumber\\
&+\int (1-I)\log(1-I)\mu(dx)\mu(dy)
\label{e80}
\end{align}
where we used $\int xy(1+\bar{x}\bar{y})\log(2xy)\mu(dx)\mu(dy)=0$ which follows from the assumptions on $\mu$.
Now if the integrand in \eqref{e80} can be expanded as $\sum_{m,n\in\{0,1,\dots\}}O_{mn}x^my^n$ then the required condition is that  $[O_{mn}]_{m,n\ge 3}$ is a positive semidefinite matrix.
We can calculate that the coefficient of $x^my^n$ in $\log\left(1-\frac{x+y-xy}{2}\right)$ is 
\begin{align}
-\sum_{k=m\vee n}^{m+n}
\frac1{k2^k}(-1)^{m+n-k}{k\choose k-n}{n\choose k-m}
\end{align}
and in $\log(1-I)$ is 
\begin{align}
(-1)^{m+n+1}\sum_{k=\lceil \frac{m\vee n}{2}\rceil}^{m\wedge n}\frac1{k}
\sum_{d=(m+n-3k)\vee 0}^{m\wedge n-k}
{k\choose d+3k-m-n}
{m+n-d-2k\choose m-k-d}
{n-k\choose d},
\end{align}
from which we easily obtain the expression of $O$. 
We numerically verified that $[O]_{2\le m,n\le L}$ is positive semidefinite for $L=29$.
For $L\ge 30$, the combinatorial numbers are not computed exactly (keeping only 15 digits), but we checked the numerical value of the minimum eigenvalue for up to $L=90$, all confirming the positive semidefinite hypothesis.
(Code can be found in \url{https://jingbol.web.illinois.edu/frankl7.m})

\subsection{9-Dimensional Optimization}\label{sec_9dim}
Under the positive semidefiniteness hypothesis (Section~\ref{sec_psd}), it is sufficient to consider distribution of the form \eqref{e_76}-\eqref{e_77}.
We can certify that $c>0$ is a lower bound for the union-closed sets conjecture if for some $\beta\in(0,1)$ the following 9-dimensional optimization has optimal value no less than 1:
\begin{align}
\textrm{minimize }\quad 
\frac{\bar{\beta}\mathbb{E}_{(\bar{q}P_0+q P_1)^{\otimes 2}}[h(XY)]
+
\beta\mathbb{E}_{\bar{q}P_0^{\otimes 2}+qP_1^{\otimes 2}}[h(XY+XY\bar{X}\bar{Y})]}
{\mathbb{E}_{\bar{q}P_0+qP_1}[h(X)]}
\end{align}
\begin{align}
\textrm{ subject to: }\quad
\bar{q}(a_1b_0+a_2b_2+a_3b_4)
+q(a_1b_1+a_2b_3+a_3b_5)&\ge 1-c;
\\
0\le a_1,a_2,q,b_0,b_1,\dots,b_5&\le 1;
\\
a_1+a_2&\le 1.
\end{align}
We used various algorithms (interior-point, sqp, active-set) from Matlab optimization package to solve constrained optimization with about $10^5$ different random initializations (code can be found in \url{https://jingbol.web.illinois.edu/frankl5.m}).
The percentage of different local optimizers found changes with the choice of solvers.
One common local minimum found was simply a one-point mass,
but the best local minimum found, which we conjecture to be global, is the following (up to symmetries):
\begin{align}
q&=0;\\
P_0&=p^*\delta_{x^*}+\bar{p}^*\delta_0,
\end{align}
where $p^*$ and $x^*$ are defined by 
\begin{align}
x^{*2}+x^{*2}(1+\bar{x}^{*2})&=1;
\\
p^{*2}h(x^{2*})-p^*h(x^*)&=0.
\end{align}
Define
$
c':=1-p^*x^*
$.
The optimal value of $\beta$ should satisfy
\begin{align}
d(\bar{\beta}p^2h(x^2)
+\beta h(p^2(1+\bar{p}^2))
-ph(x))|_{p=p^*,x=x^*}&=0;
\\
d(px)|_{p=p^*,x=x^*}&=0,
\end{align}
where the differentials are in $x$ and $p$, and we can solve $\frac{dp}{dx}$ and $\beta$ from the two equations to obtain expression of the optimal $\beta^*$ in terms of $x^*$ and $p^*$ (omitted here).
The numerical values are 
\begin{align}
p^*&\approx 0.893604513905457;\\
x^*&\approx 0.690787593924988;\\
c'&\approx 0.382709087918741;\label{e92}\\
\beta^*&\approx 0.100052559862974.
\end{align}
Our conclusion is the following:
\begin{thm}\label{thm13}
Under the positive-semidefiniteness hypothesis in Section~\ref{sec_psd} and the hypothesis of the global minimizer structure in Section~\ref{sec_9dim}, the constant in the union-closed sets conjecture can be improved to $c'$ in \eqref{e92}. 
\end{thm}

\bibliographystyle{alpha}
\bibliography{ref}

\begin{thebibliography}{LCCV18}

\bibitem[AHS22]{alweiss2022improved}
Ryan Alweiss, Brice Huang, and Mark Sellke.
\newblock Improved lower bound for the union-closed sets conjecture.
\newblock {\em arXiv preprint arXiv:2211.11731}, 2022.

\bibitem[AJN22]{anantharam2022unifying}
Venkat Anantharam, Varun Jog, and Chandra Nair.
\newblock Unifying the brascamp-lieb inequality and the entropy power
  inequality.
\newblock {\em IEEE Transactions on Information Theory}, 68(12):7665--7684,
  2022.

\bibitem[BS15]{bruhn2015journey}
Henning Bruhn and Oliver Schaudt.
\newblock The journey of the union-closed sets conjecture.
\newblock {\em Graphs and Combinatorics}, 31:2043--2074, 2015.

\bibitem[Cam22]{cambie2022better}
Stijn Cambie.
\newblock Better bounds for the union-closed sets conjecture using the entropy
  approach.
\newblock {\em arXiv preprint arXiv:2212.12500}, 2022.

\bibitem[CL22]{chase2022approximate}
Zachary Chase and Shachar Lovett.
\newblock Approximate union closed conjecture.
\newblock {\em arXiv preprint arXiv:2211.11689}, 2022.

\bibitem[Gil22]{gilmer2022constant}
Justin Gilmer.
\newblock A constant lower bound for the union-closed sets conjecture.
\newblock {\em arXiv preprint arXiv:2211.09055}, 2022.

\bibitem[Kni94]{knill1994graph}
Emanuel Knill.
\newblock Graph generated union-closed families of sets.
\newblock {\em arXiv preprint math/9409215}, 1994.

\bibitem[LCCV17]{liu2017information}
Jingbo Liu, Thomas~A Courtade, Paul Cuff, and Sergio Verdu.
\newblock Information-theoretic perspectives on {Brascamp-Lieb} inequality and
  its reverse.
\newblock {\em arXiv preprint arXiv:1702.06260}, 2017.

\bibitem[LCCV18]{liu2018forward}
Jingbo Liu, Thomas~A Courtade, Paul~W Cuff, and Sergio Verd{\'u}.
\newblock A forward-reverse {Brascamp-Lieb} inequality: Entropic duality and
  {Gaussian} optimality.
\newblock {\em Entropy}, 20(6):418, 2018.

\bibitem[Saw22]{sawin2022improved}
Will Sawin.
\newblock An improved lower bound for the union-closed set conjecture.
\newblock {\em arXiv preprint arXiv:2211.11504}, 2022.

\bibitem[W{\'o}j99]{wojcik1999union}
Piotr W{\'o}jcik.
\newblock Union-closed families of sets.
\newblock {\em Discrete Mathematics}, 199(1-3):173--182, 1999.

\bibitem[Yu23]{yu2023dimension}
Lei Yu.
\newblock Dimension-free bounds for the union-closed sets conjecture.
\newblock {\em Entropy}, 25(5):767, 2023.

\end{thebibliography}
\end{document}